\documentclass[a4paper,12pt]{article}
\usepackage[utf8]{inputenc}

\usepackage{amsmath,amssymb,amsthm,cite,eepic,epsfig,url,psfrag,tabularx,mathtools,wasysym,comment}

\usepackage{braket}
\usepackage{cancel}
\usepackage[normalem]{ulem}
\usepackage{bm}
\usepackage[a4paper, top=20mm, bottom =20mm, textwidth=165mm]{geometry}

\usepackage[utf8]{inputenc} 
\usepackage[T2A]{fontenc}
\usepackage[english]{babel}
\usepackage{xcolor,hyperref}
\usepackage{pdfpages}

\makeatletter
\def\Appendix{\appendix
  \def\@seccntformat##1{Appendix~\csname the##1\endcsname.~~}}
\makeatother
\theoremstyle{plain}
\newtheorem{theorem}{Theorem}[section]
\newtheorem{lemma}[theorem]{Lemma}
\newtheorem{fact}[theorem]{Proposition}
\newtheorem{proposition}[theorem]{Proposition}
\newtheorem{corollary}[theorem]{Corollary}
\theoremstyle{definition}
\newtheorem{define}{Definition}[section]
\newtheorem{remark}[theorem]{Remark}
\newtheorem{example}[theorem]{Example}

\renewcommand\l{\left}
\renewcommand\r{\right}
\newcommand{\ri}{\mathrm{i}}

\numberwithin{equation}{section}

\title{NSR singular vectors from Uglov polynomials}
\author{Mikhail Bershtein \and Angelina Vargulevich}
\date{}

\begin{document}

            

\maketitle 

\begin{abstract} It was conjectured in \cite{BBT2013} that bosonization of singular vector (in Neveu-Schwarz sector) of $\mathcal{N}=1$ super analog of the Virasoro algebra can be identified with Uglov symmetric function. In the paper we prove this conjecture. We also extend this result to the Ramond sector of $\mathcal{N}=1$ super-Virasoro algebra.
\end{abstract}



\section{Introduction} 

Singular vectors in Verma modules are important for the study and application of the representation theory of the Virasoro algebra and generalizations of this algebra. The existence of the singular vectors usually follows from the  formulas for determinant of the Shapovalov form, but often it is difficult to write down explicit formulas for singular vectors themselves. 

An interesting connection between singular vectors in Verma modules for Virasoro algebra and symmetric Jack functions was discovered in work \cite{mimachi1995}, and then, using different approach, in \cite{AWATA1995}. Namely, it was shown that after bosonization the singular vectors can be identified with 	\emph{Jack} symmetric functions with rectangular Young diagram. In the subsequent paper \cite{SKAO} it was proven that singular vectors for Verma modules of \(q\)-Virasoro algebra after bosonization can be identified with 	\emph{Macdonald} symmetric functions with rectangular Young diagram. In the limit \(q,t \rightarrow 1\) this result implies the Virasoro case.

In this paper we study $\mathcal{N}=1$ super analog of the Virasoro algebra. This algebra has two sectors: Neveu-Schwarz and Ramond, we will abbreviate it to NSR algebra. It was conjectured in \cite{BBT2013} that after certain special bosonization the singular vectors in NS sector are identified with Uglov symmetric functions with rectangular Young diagram. By definition, Uglov symmetric functions \cite{Uglov:1998} are limits of Macdonald polynomials when \(q,t\)  go to root of unity. In the conjecture \cite{BBT2013} the limit \(q,t\rightarrow -1\) appears. This limit was motivated by the AGT correspondence. 

In this paper we extend this conjecture to R sector and prove it in both sectors.\footnote{In the preprint \cite{Yanagida} the proof in NS sector was suggested. However that proof contains serious gaps, and we do not know how to fill them. Our proof is based on different (but related) approach.}  The proof is rather simple (using results from \cite{SKAO}, \cite{BBT2013} and \cite{Itoyama:2013}) and fills the gap in the literature. The idea of the proof is to use result \cite{SKAO} on singular vectors of \(q\)-Virasoro algebra and then take the limit. We show that in the \(q,t\rightarrow -1\) limit of the generating current \(T(z)\) of \(q\)-Virasoro algebra one gets the generating currents of the  NSR algebra with an additional fermion. 

The fact that NSR algebra appears in the limit of \(q\)-Virasoro algebra looks to be even more important than the relation on singular vectors. It was also conjectured in \cite{BBT2013}. Our computations here are very close to ones in \cite{Itoyama:2013}.  It is expected that in more general root of unity limit \(q\)-Virasoro algebra will contain certain coset or parafermion algebra, see \cite{BBT2013}, \cite{Itoyama:2014}, \cite{Kimura:2022}. 

There are another formulas for bosonization of the singular vectors of NSR algebra in terms of Super-Jack polynomials. 
Such formulas were conjectured in \cite{Desrosiers:2012} for NS sector and in \cite{Alarie-Vezina:2013} for R sector and proven in  \cite{Blondeau-Fournier:2016}. It would be interesting to understand the relation between these formulas and formulas in terms of Uglov symmetric functions. 

This paper is organized as follows. In Section \ref{sec2} we recall necessary facts on Macdonald and Uglov polynomials. In Section~\ref{sec3} we define NSR algebra and Verma modules. In Section~\ref{sec4} we recall bosonization of the NSR algebra. In Section~\ref{sec:qVir} we recall \(q\)-Virasoro algebra and prove Theorems \ref{T0T1NS}, \ref{T0T1R} which describe the limit of the algebra. In Section~\ref{sec5} we formulate the main theorem (Theorem \ref{Th:sing:NSR}) and give its proof. Finally, in Section \ref{sec6} we study first orders of the limit of Macdonald operator.

We are grateful to A. Shchechkin for
careful reading of preliminary version of the paper and many useful remarks. The work was partially supported by the HSE University Basic Research Program.

\section{Uglov Symmetric Functions} \label{sec2}
In this section we recall some basic notions about the symmetric polynomials and Macdonald symmetric functions. The reference is \cite{MacDonald1998}. 
 
A partition $\mu=\left(\mu_1,\mu_2,\ldots\right)$ is a sequence of non-negative integers in decreasing order $\mu_1 \geq \mu_2 \geq \ldots\geq 0$. The number of the non-zero $\mu_i$'s is called length of $\mu$ and denoted by  $\ell\left(\mu\right)$. By \(\left|\mu\right|\)  we denote the sum of the \(\mu_i\). The dominance order on partitions is defined as $\mu \geq \nu$ if and only if $\left|\mu\right| = \left|\nu\right|$ and $\mu_1 +\ldots+ \mu_i \geq \nu_1+\ldots\nu_i$ for any $i\geq 1$.
 
Sometimes, 
we will also use a notation 
$\mu=\left(1^{m_1},2^{m_2},\ldots\right)$, 
where $m_i$ denotes the number of times $i$ occurs in $\mu$.
 
Let $x_1,\ldots, x_N$ be independent variables. The symmetric group $S_N $ acts on the polynomial ring $\mathbb{C}\left[x_1,\ldots,x_N\right]$ by permuting the $x$'s. Denote 
\( \Lambda_N=\mathbb{C}\left[x_1,\ldots,x_N\right]^{S_N}.\)
The ring of symmetric functions on infinitely many variables \(\Lambda\) is defined as the inverse limit \(\varprojlim_N  \Lambda_N\). It is a graded ring 
\( \Lambda=\bigoplus_{r\geq 0}\Lambda^m\), where $\Lambda^m$ denotes the space of homogeneous symmetric functions of degree \(m\).

For any partition \(\mu\) by \(m_\mu \in \Lambda_N\) we denote monomial symmetric polynomial. By \(p_k\) for any \(k\in \mathbb{Z}_{>0}\) we denote \(p_k=\sum x_i^k \in \Lambda_N\). The \(N\rightarrow \infty\) limits of \(m_\mu\) and \(p_k\) are well-defined, we will denote them by the same letters \(m_\mu, p_k \in \Lambda\).
    
Let $q,t$ be independent variables and let $\mathbb{F}=\mathbb{Q}\left(q,t\right)$ be the field of rational functions on \(q,t\). Let 
\begin{equation}
	T_{q,x_i} \colon \mathbb{F}[x_1,\dots,x_N] \rightarrow \mathbb{F}[x_1,\dots,x_N], \quad T_{q,x_i}(x_j)=q^{\delta_{i,j}}x_j
\end{equation}
denotes the shift operator. The \emph{Macdonald difference operator} is defined by the formula
\begin{equation}\label{eq:Macd:oper:N}
	D_{q, t}=\sum_{i=1}^{N} \Big( \prod\limits_{j=1 \atop j \neq i}^{N} \frac{t x_{i}-x_{j}}{x_{i}-x_{j}}\Big) T_{q, x_{i}}.
\end{equation}
It can be shown that \(D_{q, t}\) preserves the space of symmetric polynomials \(\Lambda_N\). The next theorem (see \cite{MacDonald1998} for the proof) defines Macdonald symmetric polynomials.

\begin{theorem}\label{Th:Macd:polyn}
	For any partition $\mu$, $\left(\ell(\mu)\leq N \right)$ there is a unique symmetric polynomial $P_{\mu}\left(q,t\right)\in \Lambda_{N,\mathbb{F}}$ such that 
     \begin{itemize}
            \item $D_{q,t}P_{\mu}\left(q,t\right)=\mathcal{E}_{\mu,N}\left(q,t\right)P_{\mu}\left(q,t\right)$, where \(\mathcal{E}_{\mu,N}\l(q,t\r)=\sum\limits_{j=1}^{N} q^{\mu_j}t^{N-j}\),
            \item               $P_{\mu}\l(q,t\r)=m_{\mu}+\sum\limits_{\nu<\mu}u_{\mu \nu }\l(q,t\r)m_{\nu}$, where $u_{\mu \nu }\l(q,t\r)\in \mathbb{F}$.
    \end{itemize}
\end{theorem}
Macdonald symmetric function \(P_\mu(q,t)\in \Lambda_\mathbb{F}\) is the limit of Macdonald symmetric polynomial, when the number of variables \(N\)
goes to infinity.

\begin{lemma} \label{Maclem}
	Macdonald symmetric function $P_{\mu}\l(q,t\r)$ has a property 
    \begin{equation}
    	P_{\mu}\l(q,t\r)=P_{\mu}\l(q^{-1},t^{-1}\r).
    \end{equation}
\end{lemma}

See \cite[Ch. VI (4.14) iv]{MacDonald1998} for the proof of the lemma.

\begin{define}
	\emph{Uglov symmetric function} \(P_\mu^{(\gamma,l)} \in \Lambda_N\) is the limit 
	\begin{equation}\label{eq:Uglov-def}
		P^{(\gamma,l)}_{\mu}=\lim\limits_{q\rightarrow 1}P_{\mu}\l(\omega_l q, \omega_l q^{\gamma}\r),
	\end{equation}
    where $l\in \mathbb{Z}_{>0}$, $\gamma\in \mathbb{C}$ and $\omega_l=\exp\left(2\pi \ri/l\right)$.
\end{define}
These functions were introduced in \cite{Uglov:1998} under the name \(\mathfrak{gl}_l\)-Jack polynomials. For \(l=1\) they are usual Jack polynomials. In the paper, we use only Uglov polynomials for \(l=2\). 

It is not difficult to see that the limit \eqref{eq:Uglov-def} exists, see e.g. \cite[App. B.2]{BBT2013}.

\section{NSR algebra}
\label{sec3}

In this section, we recall the definition \(\mathcal{N}=1\) super analog of the Virasoro algebra and its Verma modules. Another name of this algebra is Neveu–Schwarz-Ramond algebra, or just NSR algebra for brevity. 

\begin{define}
\emph{NSR algebra} is a Lie super algebra with generators $L_n$,  $G_k$ and with central generator $c$ subject of the following relations:
\begin{subequations}
    \begin{align}
	    &\left[L_{m}, L_{n}\right]=(m-n)L_{m+n} + \frac{c}{12}\left(m^{3}-m\right) \delta_{m+n, 0},  \label{NSR1}\\
	    &\left[L_{n},G_{k}\right] = \left(\frac{n}{2}-k\right) G_{n+k} , \label{NSR2}\\
	    &\left\{G_{k}, G_{l}\right\} = 2L_{k+l} + \frac{c}{3}\left(k^{2}-\frac{1}{4}\right) \delta_{k+l,0} \label{NSR3},
    \end{align}
\end{subequations}
where $m,n \in \mathbb{Z} $ and  $k,l \in \mathbb{Z}+\delta $. The case of $\delta=1/2 $ is called Neveu-Schwarz (NS) sector of the NSR algebra, and the case of $\delta =0 $ is called Ramond (R) sector of the NSR algebra.
\end{define}

In all representations considered in the paper, \(c\) acts as a complex number, which we denote by the same letter \(c \in \mathbb{C}\). This number is called central charge. 

\begin{lemma}
    In the highest weight representations of central charge \(c\) the relations of the NSR algebra can be written in form of the operator product expansions as
	\begin{subequations}
	    \begin{align}
	    &L(z) L(w)=\frac{c}{2(z-w)^{4}}+\frac{2 L(w)}{(z-w)^{2}}+\frac{L^{\prime}(w)}{z-w}+\text{reg}
	     ,\\
	    &L(z) G(w)=\frac{3 G(w)}{2(z-w)^{2}}+\frac{G^{\prime}(w)}{z-w}+\text{reg} ,\\
	    &G(z) G(w)=\frac{2 c}{3(z-w)^{3}}+\frac{2 L(w)}{z-w}+\text{reg}.
	    \end{align}
	\end{subequations}
	where $L(z)=\sum\limits_{n\in \mathbb{Z}}L_nz^{-n-2}$, $G(z)=\sum\limits_{k\in \mathbb{Z + \delta}}G_k z^{-k-3/2}$ and ``$\text{reg}$'' stands for terms regular at $z\rightarrow w$.
\end{lemma}

It is convenient to parametrize the central charge as
\begin{equation}\label{eq:cNSR}
    c=\frac{3}{2}-12\rho^2, \quad \rho=\frac{1}{2}\left(\beta - \frac{1}{\beta}\right) .
\end{equation}

\begin{define}[NS]
	Verma module of NSR algebra in NS sector $M\left(c,\Delta\right)$ is freely generated by \(L_{-n},G_{-k}\),  \(\forall n,k>0\) from the highest weight vector $\ket{\Delta}$ defined by 
\begin{equation}
        L_n\ket{\Delta}=0,\ (\forall n\in \mathbb{Z}_{>0}) \quad L_0\ket{\Delta}=\Delta\ket{\Delta} \quad G_k\ket{\Delta}=0,\ (\forall k\in \mathbb{Z}_{\geq 0}+\frac{1}{2}).
    \end{equation}
\end{define}
\begin{define}[R]
\label{VermaR}
Verma module of NSR algebra in R sector $M\left(c,\lambda\right)$ is freely generated by \(L_{-n},G_{-k}\), \(\forall n,k>0\) from the highest weight vector  \(\ket{\Delta}=\ket{\Delta,\lambda}\) defined by 
\begin{subequations}
	\begin{align} \label{RVermaa}
        L_n\ket{\Delta}&=0,\ (\forall n\in \mathbb{Z}_{>0}) \quad L_0\ket{\Delta}=\Delta\ket{\Delta} \quad G_k\ket{\Delta}=0,\ (\forall k\in \mathbb{Z}_{>0}),\\
		\label{RVermaa:G0}
    	G_0\ket{\Delta}&=\lambda\ket{\Delta}.
	\end{align}    
\end{subequations}
Here $\Delta=\lambda^2+c/24$.
\end{define}

\begin{remark}\label{Rem:M:tilde}
	There is another definition of the Verma module in the R sector (see e.g. \cite{Iohara:20031}). In this definition the module is generated by the highest weight vector subject of relations \eqref{RVermaa} (without \eqref{RVermaa:G0}). We denote this module by \(\widetilde{M}(c,\Delta)\). It has two highest weight vectors: even $\ket{\Delta^+}$ and odd $\ket{\Delta^-}$ such that
	\begin{subequations}
		\begin{align}
		    L_n\ket{\Delta^{\pm}}&=0, \;(\forall n\in \mathbb{Z}_{>0}), \quad L_0\ket{\Delta^{\pm}}=\Delta\ket{\Delta^{\pm}}, \quad G_k\ket{\Delta^{\pm}}=0, \; \left(k>0, k\in \mathbb{Z}\right),
		    \\
		    G_0\ket{\Delta^+}&=\ket{\Delta^-}, \quad G_0\ket{\Delta^-}=(\Delta-c/24)\ket{\Delta^+}.
		\end{align}
	\end{subequations}
		Assume that \(\Delta \neq c/24\) and pick \(\lambda=\sqrt{\Delta-c/24}\). Then one can define vectors \(\ket{\Delta,\lambda}\) and \(\ket{\Delta,-\lambda}\) by the formulas
	\begin{equation} \label{altR}
	    \ket{\Delta^+}=\frac{\ket{\Delta,\lambda}-\ket{\Delta,-\lambda}}{2\lambda},\quad \ket{\Delta^-}=\frac{\ket{\Delta,\lambda}+\ket{\Delta,-\lambda}}{2}.
	\end{equation}
	The vectors \(\ket{\Delta,\lambda}\) and \(\ket{\Delta,-\lambda}\) satisfy conditions \eqref{RVermaa}-\eqref{RVermaa:G0} and generate modules \(M(c,\lambda)\) and \(M(c,-\lambda)\). Hence, we get (under assumption \(\Delta \neq c/24\)) \(\widetilde{M}(c,\Delta)=M(c,\lambda)\oplus M(c,-\lambda)\).
\end{remark}
\begin{define}
	The vector \(\chi\) in the Verma module is called singular if 
	\begin{equation}\label{eq:sing}
	    L_m\ket{\chi}=G_k\ket{\chi}=0,\ \forall k>0, m>0.
	\end{equation}
\end{define}

Without loss of generality, one can assume that the singular vector is an eigenvector of \(L_0\), namely 
\begin{equation}\label{eq:L0:chi}
	L_0\chi = (\Delta+n)\chi, \quad n \in \mathbb{Z}\cup \mathbb{Z}+\delta
\end{equation}
The value \(n\) is called the level of the singular vector.

The Verma module is an irreducible representation for generic values of the highest weight. However, for special values of the highest weight and generic central charge Verma module is a reducible representation and it has only one singular vector. These special values are given by the analog of Kac-Feigin-Fuchs theorem for NSR algebra (see e.g. \cite[App. 1]{Kac:1986}).

\begin{theorem}
	The Verma module has a singular vector if 
	\begin{equation}\label{eq:Deltars}
		\Delta=\Delta_{r,s}=\frac{1}{8}\left(r\beta-s\beta^{-1}\right)^2-\frac{\rho^2}{2}+\frac{1-2\delta}{16},
	\end{equation}
	and 
	\begin{equation}\label{lamrs}
		\lambda=\lambda_{r,s}=\varepsilon \frac{r\beta-s\beta^{-1}}{2\sqrt{2}},
	\end{equation}
	for R sector, where $\varepsilon$ is a sign. Here $r,s\in \mathbb{Z}_{>0}$ such that $\left(r-s\right) \operatorname{mod} 2 =1-2\delta$ and level of the singular vector is \(n= r s/2\).
%
\end{theorem}

The eigenvalue of $G_0$ in the formula \eqref{lamrs} for R sector is derived from the value of \(\Delta\) in \eqref{eq:Deltars} using commutation relation \eqref{NSR3}. 

\begin{fact}\label{sign}
	In the R sector the singular vector \(\chi\) satisfies 
	\begin{equation}\label{eq:G0:chi}
		G_0 \chi = \tilde{\varepsilon}\frac{r\beta+s\beta^{-1}}{2\sqrt{2}}\chi
	\end{equation}
	where $\tilde{\varepsilon}=(-1)^s\varepsilon$.
%
%
%
	\end{fact}
	Up to sign the formula \eqref{eq:G0:chi} follows from
	\eqref{eq:L0:chi}. For the sign see \cite[eq. 2.10]{Watts_1993} which is based on \cite{Friedan:1988}. We will reproduce this sign below.

\section{Bosonization} \label{sec4}

\begin{define}\label{def:HeisCliff:algebra}
	Heisenberg-Clifford algebra is an algebra with generators $a_n$, $\forall n\in \mathbb{Z}$ and $f_k$ $\forall k\in\mathbb{Z}+\delta$ and relations
	\begin{equation}
	    \l[a_n,a_m\r]=n\delta_{m+n,0},\quad \l\{f_k,f_l\r\}=\delta_{k+l,0}, \quad \l[a_n,f_k\r]=0.
	\end{equation}
\end{define}

Recall that \(\delta=1/2\) corresponds to Neveu-Schwarz sector and \(\delta=0\) corresponds to Ramond sector. 

\begin{lemma} 
	The relations of the Heisenberg-Clifford algebra can be written in the form of the operator product expansions, 
    \begin{equation} \label{af}
        a(z)a(w)=\frac{1}{(z-w)^2}+\text{reg}, \quad f(z)f(w)=\frac{1}{z-w}+\text{reg},
    \end{equation}
     where $a(z)=\sum\limits_{n\in\mathbb{Z}}a_n z^{-n-1}$ is a bosonic field and  $f(z)=\sum\limits_{k\in\mathbb{Z}+\delta} f_k z^{-k-1/2}$ is a fermionic field.
\end{lemma}

\begin{define}[NS] \label{def:Fock:NS}
	Fock module of Heisenberg-Clifford algebra in NS sector $F_{\alpha}$, \(\alpha \in \mathbb{C}\) is freely generated by \(a_{-n},f_{-k}\), \(\forall n,k>0\) from the highest weight vector $\ket{\alpha}$ defined by 
    \begin{equation}
        a_n\ket{\alpha}=0\ (\forall n>0), \quad a_0\ket{\alpha}=\alpha\ket{\alpha}, \quad f_k\ket{\alpha}=0\ (\forall k>0, k\in \mathbb{Z}+1/2).
    \end{equation}
\end{define}
\begin{define}[R] \label{def:Fock:R}
	Fock module of Heisenberg-Clifford algebra in R sector $F_{\alpha}=F_{\alpha,\varepsilon }$, \(\alpha \in \mathbb{C}, \varepsilon=\pm 1\) is freely generated by \(a_{-n},f_{-k}\), \(\forall n,k>0\) from the highest weight vector $\ket{\alpha}=\ket{\alpha,\varepsilon}$ defined by 
\begin{subequations}
    \begin{equation}
        a_n\ket{\alpha}=0\ (\forall n\in \mathbb{Z}_{>0}), \quad a_0\ket{\alpha}=\alpha\ket{\alpha}, \quad f_k\ket{\alpha}=0\ (\forall k \in \mathbb{Z}_{>0}),
    \end{equation}
    \begin{equation}
    f_0\ket{\alpha}=\frac{\varepsilon}{\sqrt{2}}\ket{\alpha}.
    \end{equation}
    \end{subequations}
\end{define}
The following lemma is a standard bosonization of the NSR algebra.    
\begin{lemma} Formulas 
\label{vir}
    \begin{equation} 
    \begin{aligned}
    L(z):=\sum_{n \in Z} L_{n} z^{-n-2} \quad & \longmapsto \frac{1}{2} : a(z)^{2} :+\rho\partial_z a(z)+\frac{1}{2} :\!\!\partial_z f(z) f(z)\!\!:, \\
    G(z):=\sum_{k \in Z+\delta} G_{k} z^{-k-3 / 2} & \longmapsto f(z) a(z)+2 \rho\, \partial_z f(z)
    \end{aligned}
    \end{equation}
    define the action of the NSR algebra on Fock module $F_{\alpha}$. 
\end{lemma}
The map defines representations homomorphism $M\left(c,\Delta\right) \rightarrow F_{\alpha}$ in NS sector
\begin{equation}
    \ket{\Delta}\mapsto \ket{\alpha}, \quad \Delta = \frac{1}{2}\left(\alpha^2-2\rho\alpha\right)+\frac{1-2\delta}{16}, 
\end{equation}
and $M\left(c,\lambda\right)\rightarrow F_{\alpha,\varepsilon}$ in R sector
\begin{equation}
    \ket{\Delta, \lambda}\mapsto \ket{\alpha,\varepsilon}, \quad \lambda=\varepsilon\frac{\alpha-\rho}{\sqrt{2}}.
\end{equation}
Special values of the highest weight vector $\Delta=\Delta_{r,s}$ correspond to the special values of the moments  $\alpha=\alpha_{r,s}$
\begin{equation}\label{eq:alpha:rs}
    \alpha_{r,s}=\frac{1}{2}\left(1+r\right)\beta-\frac{1}{2}\left(1+s\right)\beta^{-1}.
\end{equation}

\subsection{Odd bosonization}

In order to state our result, it is convenient to bosonize fermion current \(f(z)\) in terms of additional boson. This boson appears to be odd. Such formulas can be viewed as a version of boson-fermion correspondence. See for example \cite[Sec. 2.3]{Itoyama:2013} or \cite[App. B]{Bershtein2018}. 

It will be also necessary to extend Heisenberg-Clifford algebra by an additional fermion. 
Similarly to Definition \ref{def:HeisCliff:algebra} one can define Heisenberg-Clifford-Clifford algebra with generators \(a_n,f_k,f_r\), with \(n,k\in \mathbb{Z}, r \in \mathbb{Z}+1/2\). This algebra has two fermionic currents, which we denote by $f^{NS}(z)=\sum\limits_{r\in\mathbb{Z}+1/2} f_r z^{-r-1/2}$ and $f^{R}(z)=\sum\limits_{k\in\mathbb{Z}} f_k z^{-k-1/2}$.
\begin{lemma} \label{iso}
Formulas 
    \begin{equation}   \label{isofns}
		f^{NS}(z^2) \longmapsto \frac{\ri}{2\sqrt{2} z}\left(e^{\varphi_{-}(z)} e^{ \varphi_{+}(z)}-e^{-\varphi_{-}(z)} e^{-\varphi_{+}(z)}\right),
	\end{equation}
	\begin{equation} \label{isofr}
		f^{R} (z^2) \longmapsto \frac{\varepsilon}{2\sqrt{2} z}\left(e^{\varphi_{-}(z)} e^{ \varphi_{+}(z)}+e^{-\varphi_{-}(z)} e^{- \varphi_{+}(z)}\right),
	\end{equation}
	where
	\begin{equation}\label{phi+-}
		\varphi_{-}(z)=-\sum_{n \in \mathbb{Z}_{>} 0} 	\frac{p_{2 n-1}}{2 n-1} z^{2n-1},
		\quad \quad \varphi_{+}(z)=\sum_{n \in \mathbb{Z}_{>0}} 2\frac{\partial}{\partial p_{2 n-1}} z^{-2n+1},
	\end{equation}
	and
	\begin{equation}
		\label{isoan}
	    a_n \mapsto -2\beta n\frac{\partial}{\partial 	p_{2n}}, \quad \quad a_{-n} \mapsto -\frac{1}{2\beta}p_{2n},\quad a_0 \mapsto \alpha  
	\end{equation}
	define the action of Heisenberg-Clifford-Clifford algebra on the space of symmetric functions~$\Lambda$.
\end{lemma}
\begin{proof}
    For Heisenberg algebra, relations $\left[a_n,a_m\right]=n\delta_{n+m,0}$ are satisfied. For fermions, we have to check the defining relations \eqref{af}. It is easy to see that
    \begin{equation}
        \left[\varphi_+(z),\varphi_-(w)\right]=\ln{\frac{1-w/z}{1+w/z}}.
    \end{equation}
    Then straightforward computation gives 
	\begin{equation}
	  \left\{f^{NS}\left(z^2\right),f^{NS}\left(w^2\right)\right\}=\frac{1}{4z^2}\delta\left(\frac{w^2}{z^2}\right), \quad \left\{f^R\left(z^2\right),f^R\left(w^2\right)\right\}=\frac{1}{4zw}\delta\left(\frac{w^2}{z^2}\right),
	\end{equation}
	where \(\delta(x)=\sum_{n \in \mathbb{Z}} x^n\).
\end{proof}

Similarly to Definitions \ref{def:Fock:NS},\ref{def:Fock:R} we define Fock module over Heisenberg-Clifford-Clifford algebra. We denote this module by \(\mathrm{F}_{\alpha,\varepsilon}\).

\begin{proposition}\label{prop:Lambda:HCC}
	Lemma \ref{iso} defines isomorphism \(\Lambda \cong \mathrm{F}_{\alpha,\varepsilon}\).
\end{proposition}
\begin{proof}
	The Lemma~\ref{iso} defines a map \(\mathrm{F}_{\alpha,\varepsilon}\ \rightarrow \Lambda\). Since \(\mathrm{F}_{\alpha,\varepsilon}\) is irreducible, this map is injective. The surjectivity will follow from the equality of the characters (i.e. Hilbert–Poincaré series) for both sides. 
	
	For \(\Lambda\) we have 
	\begin{equation}
		\operatorname{ch}(\Lambda)=\sum_{m \geq 0} \mathfrak{q}^m \dim \Lambda^m=\prod_{k\geq 1} \frac{1}{1-\mathfrak{q}^k}
	\end{equation}
	
	The natural grading on the Heisenberg-Clifford-Clifford algebra is defined as follows 
	\begin{equation}
		\deg f_k= -2k,\;\; \deg f_r =-2r,\;\; \deg a_m=-2m\quad k,m\in \mathbb{Z}, r\in \mathbb{Z}+\frac{1}{2}.
	\end{equation}
	
 	Assuming that \(\deg \ket{\alpha,\varepsilon}=0\) we get the character of the Fock module
 	\begin{multline}
 		\operatorname{ch}(\mathrm{F}_{\alpha,\varepsilon})=\prod_{m \in \mathbb{Z}_{>0}} \frac{1}{1-\mathfrak{q}^{2m}}\prod_{k \in \mathbb{Z}_{>0}} (1+\mathfrak{q}^{2k}) \prod_{r+1/2 \in \mathbb{Z}_{>0}} (1+\mathfrak{q}^{2r}) \\
 		= \prod_{m \in \mathbb{Z}_{>0}} \frac{1}{1-\mathfrak{q}^{2m}}\prod_{k \in \mathbb{Z}_{>0}} (1+\mathfrak{q}^{k}) 
 		=\prod_{m \in \mathbb{Z}_{>0}} \frac{1}{1-\mathfrak{q}^{2m}}\prod_{m \in \mathbb{Z}_{>0}} \frac{1}{1-\mathfrak{q}^{2m-1}}=		\operatorname{ch}(\Lambda)
 	\end{multline}
 	Here we used Euler identity (particular case of the Glaisher theorem).
\end{proof}

In this paper we consider both Neveu-Schwarz sector and Ramond sector. So, one of the fermions is an additional, namely in NS sector we denote $f^{add}=f^R$, and in R sector we denote $f^{add} = f^{NS}$.  By $F^{add}$ we denote the Fock module for additional fermion (we omit dependence on \(\varepsilon\) here). Clearly, we have \(\mathrm{F}_{\alpha,\varepsilon}=F_{\alpha}\otimes F^{add}\). 

\begin{corollary}\label{Cor:iso}
	Lemma \ref{iso} defines isomorphism \(\Lambda \cong F_{\alpha}\otimes F^{add}\).
\end{corollary}

\section{\textit{q}-Virasoro algebra}\label{sec:qVir}

In this section we define \(q\)-deformed Virasoro algebra, its Fock modules and singular vectors following \cite{SKAO}. Then we consider the limit of $Vir_{q,t}$ in the Fock module. 

\begin{define}
    The \(q\)-deformed Virasoro algebra $Vir_{q,t}$ is an associative algebra generated by $T_n$, $n\in \mathbb{Z}$ with following relations:
\begin{equation}
        \left[T_{n}, T_{m}\right]=-\sum_{l=1}^{\infty} \mathrm{f}_{l}\left(T_{n-l} T_{m+l}-T_{m-l} T_{n+l}\right)-\frac{(1-q)\left(t-1\right)}{t-q}\left(\left(\frac{q}{t}\right)^{n}-\left(\frac{q}{t}\right)^{-n}\right) \delta_{m+n, 0},
    \end{equation}
where the coefficients \(\mathrm{f}_l\)'s are given by the following function \(\mathrm{f}(z)\)
\begin{equation}
    \mathrm{f}(z)=\sum\limits_{l=0}^{\infty}\mathrm{f}_lz^l=\exp\left\{\sum\limits_{n=1}^{\infty}\frac{1}{n}\frac{\left(1-q^n\right)\left(t^n-1\right)}{\left(t^n+q^n\right)}z^n\right\}.
\end{equation}
\end{define}

The next theorem is proven in \cite[Sec. 4]{SKAO}.

\begin{theorem}\label{Th:Virqt:boson}
    The formula 
    \begin{equation}
        \begin{aligned}
        \label{stress-energy}
        &T(z)=q^{1 / 2}t^{-1/2} \exp \left\{-\sum_{n=1}^{\infty} \frac{1-t^{n}}{t^n+q^{n}}\frac{t^{n/2}}{q^{n/2}} \frac{p_n}{n} z^{n}  \right\} \exp \left\{-\sum_{n=1}^{\infty}\left(1-q^{n}\right)\frac{q^{n/2}}{t^{n/2}} \frac{\partial}{\partial p_n} z^{-n} \right\} u\\
        &+q^{-1 / 2}t^{1/2} \exp \left\{\sum_{n=1}^{\infty} \frac{1-t^{n}}{t^n+q^{n}}\frac{q^{n/2}}{t^{n/2}} \frac{p_n}{n} z^{n} \right\} \exp \left\{\sum_{n=1}^{\infty}\left(1-q^{n}\right)\frac{t^{n/2}}{q^{n/2}} \frac{\partial}{\partial p_n} z^{-n} \right\} u^{-1},
        \end{aligned}
    \end{equation}
	where \(T(z)=\sum\limits_{n\in \mathbb{Z}} T_n z^{-n}\) defines representations of \(Vir_{q,t}\) in the space \(\Lambda\).
\end{theorem}

The representation constructed in Theorem \ref{Th:Virqt:boson} depends on \(q,t\) which are parameters of \(Vir_{q,t}\) and also on \(u\). The next theorem is proven in \cite[Sec. 5]{SKAO}.

\begin{theorem}   \label{Macrs}
    For $u=u_{r,s}=t^{\frac{1}{2}\left(1+s\right)}q^{-\frac{1}{2}\left(1+r\right)}$, $r,s\in \mathbb{Z}_{>0}$ the module constructed in Theorem~\ref{Th:Virqt:boson} has singular vector $\chi_{r,s}=P_{\left(r^s\right)}\left(q,t\right)$
    \begin{equation}\label{eq:Tn:Macdonald}
        T_n P_{\left(r^s\right)}\left(q,t\right)=0,\quad n>0.
    \end{equation}
\end{theorem}

Here \((r^s)\) denotes partition that consists of \(s\) parts equal to \(r\). The corresponding Young diagram is a rectangle \(r \times s\).
\subsection{Limit}
Consider a limit  \(q=-e^{\hbar}\), \(t=-e^{\gamma \hbar}\), \(\hbar\rightarrow 0\) of the current \(T(z)\) given by the formula~\eqref{stress-energy}.
\begin{equation}
   q^{-1/2}t^{1/2}u^{-1} T(z)=T^0(z)+T^1(z)\hbar+O\l(\hbar^2\r).
\end{equation}
We will study two cases of the limit behavior of \(u\). We set \(\gamma=\beta^{-2}\), after the limit the parameter \(\beta\) will parameterize central charge of the NSR algebra via the formulas \eqref{eq:cNSR}. The following results is similar to \cite[Sec. 2.4]{Itoyama:2013}.
\begin{theorem}[NS] \label{T0T1NS}
    Let $u^2\longrightarrow 1-2\hbar\beta^{-1}\alpha +O\left(\hbar^2\right)$ Then the operators $T^0(z)$ and $T^1(z)$ are expressed through Ramond fermion $f^{add}(z)=f^R(z)$ and NSR current $G(z)$
    \begin{subequations}
	    \begin{align}
    	    \label{T0NS}
        	T^0(z)&=\varepsilon 2\sqrt{2}z f^{add}\left(z^2\right),
	        \\
    	    \label{T1NS}
        	T^1(z)&=-\ri2\sqrt{2}\beta^{-1}z^3G\l(z^2\r)-\varepsilon \sqrt{2}z^2\partial_z f^{add}\l(z^2\r)-\varepsilon \sqrt{2}f^{add}\l(z^2\r)\l(2-2{\alpha}\beta^{-1}-\beta^{-2}\r).
	    \end{align}
	\end{subequations}
\end{theorem}
\begin{theorem}[R] \label{T0T1R}
	Let  $u^2 \rightarrow -\left(1-2\hbar\beta^{-1}\alpha +O\left(\hbar^2\right)\right)$ Then the operators $T^0(z)$ and $T^1(z)$ are expressed through Neveu-Schwarz fermion $f^{add}(z)=f^{NS}(z)$ and NSR current $G(z)$
	\begin{subequations}
		\begin{align}
	        \label{T0R}
	        T^0(z)&=\ri 2\sqrt{2}z f^{add}\left(z^2\right),
	        \\
	        \label{T1R}
	        T^1(z)&=-\varepsilon 2\sqrt{2}\beta^{-1} z^3G\l(z^2\r)-\ri \sqrt{2}z^2\partial_z f^{add}\l(z^2\r)-\ri \sqrt{2}f^{add}\l(z^2\r)\l(2-2\alpha \beta^{-1}-\beta^{-2}\r).
    \end{align}
	\end{subequations}
\end{theorem}
Note that in the formula \eqref{T1NS} we used bosonization of \(G\left(z\right) \) through fermion \(f^{NS}\left(z\right)\) and on the contrary in the formula \eqref{T1R} bosonization of \(G\left(z\right) \) is through fermion \(f^R\left(z\right)\). The parameter \(\alpha\) after the limit will coincide with zero mode \(a_0\) on the Fock module \(\mathrm{F}_{\alpha,\varepsilon}\) see Proposition~\ref{prop:Lambda:HCC}.
\begin{proof}
   Let us fix $u^2 \longrightarrow \pm\left(1-2\hbar\beta^{-1}\alpha +O\left(\hbar^2\right)\right)$ for $\hbar\rightarrow 0$, where $+$ for NS sector and $-$ for R sector. Here and below upper sign corresponds to NS sector and lower sing to R sector. In zeroth order of the \(\hbar\) expansion of a $q^{-1/2}t^{1/2} u^{-1} T\l(z\r)$ we get 
   \begin{equation}
       T^0\left(z\right)=\exp \left(-\varphi_-(z)\right) \exp \left(-\varphi_+(z)\right)\pm 
    \exp \left(\varphi_-(z)\right) \exp \left(\varphi_+(z)\right).
   \end{equation}
   Due to bosonization formulas \eqref{isofns}, \eqref{isofr} we get  \eqref{T0NS}  and \eqref{T0R}. In the first order of the \(\hbar\) expansion we have 
    \begin{multline}
     T^1(z)= \frac{\gamma}{2}z\partial_z\l[ \exp \left(-\varphi_-(z)\right) \exp \left(- \varphi_+(z)\right)\mp \exp \left(\varphi_-(z)\right) \exp \left( \varphi_+(z)\right)\r]+\\
   +\l[ \exp \left(-\varphi_-(z)\right) \exp \left(- \varphi_+(z)\right)\mp \exp \left(\varphi_-(z)\right) \exp \left(\varphi_+(z)\right)\r]\l(-z^2\frac{\gamma}{2}a_+(z)-2z^2 a_-(z)\r)-\\
   -z\partial_z\l[\exp \left(-\varphi_-(z)\right) \exp \left(-\varphi_+(z)\right)\r]\mp \exp \left(\varphi_-(z)\right) \exp \left(\varphi_+(z)\right)\l(1-2{\alpha}\beta^{-1}-\gamma\r).
    \end{multline}
    Here \(a_+\left(z\right)=\sum\limits_{n>0} - p_{2n}z^{2n-2}\), \(a_-\left(z\right)=\sum\limits_{n>0}-n\frac{\partial}{\partial p_{2n}}z^{-2n-2}\) and bosonic field is  \(a\left(z^2\right)=2\beta a_-\left(z\right)+\frac{1}{2\beta}a_+\left(z\right)+\alpha z^{-2}\). Substituting $\gamma=1/\beta^2$ and using odd bosonization \ref{iso} we get
    \begin{multline}
     T^1(z)= \frac{\beta^{-2}}{2}z\partial_z\l[ \ri 2 \sqrt{2}zf^{NS}\left(z^2\right)\r]
   -\l[\ri 2 \sqrt{2}zf^{NS}\left(z^2\right) \r]\beta^{-1}z^2\l(a\l(z^2\r)-\alpha z^{-2}\r)-\\
   -z\partial_z\l[\ri  \sqrt{2}zf^{NS}\left(z^2\right)+\varepsilon  \sqrt{2}z f^R\left(z^2\right)\r]- \l(\varepsilon  \sqrt{2} z f^R\l(z^2\r)-\ri  \sqrt{2}zf^{NS}\left(z^2\right)\r)\l(1-2{\alpha}\beta^{-1}-\beta^{-2}\r)
    \end{multline}
    for NS sector and 
   \begin{multline}
     T^1(z)= \frac{\beta^{-2}}{2}z\partial_z\l[ \varepsilon 2 \sqrt{2}zf^{R}\left(z^2\right)\r]
   -\l[\varepsilon 2 \sqrt{2}zf^{R}\left(z^2\right) \r]\beta^{-1}z^2\l(a\l(z^2\r)-\alpha z^{-2}\r)-\\
   -z\partial_z\l[\ri  \sqrt{2}zf^{NS}\left(z^2\right)+\varepsilon  \sqrt{2}z f^R\left(z^2\right)\r]+\l(\varepsilon  \sqrt{2} z f^R\l(z^2\r)-\ri  \sqrt{2}zf^{NS}\left(z^2\right)\r)\l(1-2{\alpha}\beta^{-1}-\beta^{-2}\r)
    \end{multline}
    for R sector. After straightforward calculation we get formulas \eqref{T1NS} and \eqref{T1R}.
    \end{proof}
	
	Now assume that \(u=u_{r,s}\) as in Theorem \ref{Macrs}. If \(r,s\) have same parity then we are under conditions of the Theorem \ref{T0T1NS}, if \(r,s\) have different parity then we are under conditions of the Theorem \ref{T0T1R}. In any case we have \(\alpha_{r,s}=\beta(r+1)/2-\beta^{-1}(s+1)/2\) in agreement with the formula \eqref{eq:alpha:rs}.

\section{Singular vectors and symmetric polynomials}
\label{sec5}
Now we formulate the main result of this paper.
\begin{theorem}\label{Th:sing:NSR}
	For \(\alpha\) given by \eqref{eq:alpha:rs}, the NSR singular vector in \(F_{\alpha}\) (or \(F_{\alpha,\epsilon}\) in R sector) maps to Uglov symmetric functions $P_{\l(r^s\r)}^{1/\beta^2,2}$ under the map in Corollary \ref{Cor:iso}.
\begin{equation}
      \chi_{r,s}\mapsto \operatorname{const} P_{\l(r^s\r)}^{1/\beta^2,2}\in F_{\alpha}\otimes \ket{1}
\end{equation}
\end{theorem}	
\begin{remark}\label{Rem:Yanagida}
	In the NS sector this fact was stated as a conjecture in \cite{BBT2013}. The R sector is new (but rather straightforward analog of the NS sector). 
	
	In the preprint \cite{Yanagida} the proof in NS sector was suggested. The idea was similar to the idea of the proof of Theorem \ref{Macrs} in \cite{SKAO}, namely to show that singular vector are eigenfunctions of certain operators (limit of Macdonald operator). But the proof in \cite{Yanagida} contains serious gaps and we do not know how to fill them. In particular, the operators (limit of Macdonald operator) used in loc. cit. has highly degenerate spectrum, as we will see in Sec. \ref{sec6}.
	
	Here we use different (but related) approach, namely we take the limit of the Theorem \ref{Macrs} itself (not its proof).
\end{remark}

\begin{example}[NS sector] Consider  the highest weights given by formula \eqref{eq:Deltars} with \(rs\leq 3\)
	\begin{equation}
    	\Delta_{1,1}=0,\quad \Delta_{3,1}=\beta^2-\frac{1}{2},\quad \Delta_{1,3}=\frac{1}{\beta^2}-\frac{1}{2}.
	\end{equation}
	The corresponding singular vectors in Verma modules has the form
	\begin{subequations}
	\begin{align}
    	\chi_{1,1}&=G_{-1/2}\ket{\Delta_{1,1}},\\ 
	    \chi_{3,1}&=\l(L_{-1}G_{-1/2}-\beta^2G_{-3/2}\r)\ket{\Delta_{3,1}}, \\
	    \chi_{1,3}&=\l(L_{-1}G_{-1/2}-\beta^{-2}G_{-3/2}\r)\ket{\Delta_{3,1}}.
	\end{align}
	\end{subequations}
	After bosonization by Lemma \ref{vir} we get 
	\begin{subequations}
		\begin{align}
			\chi_{1,1}&\mapsto \l(\beta+\beta^{-1}\r)f_{-1/2}\ket{\alpha_{1,1}},
			\\
			\chi_{3,1}&\mapsto \l(\l(4\beta-3\beta^3-\beta^{-1}\r)f_{-3/2}+\l(3\beta^2-4+\beta^{-2}\r)a_{-1}f_{-1/2}\r)\ket{\alpha_{3,1}},
			\\
			\chi_{1,3}&\mapsto \l(\l(\beta-4\beta^{-1}+3\beta^{-2}\r)f_{-3/2}+\l(\beta^2-4+3\beta^{-1}\r)a_{-1}f_{-1/2}\r)\ket{\alpha_{1,3}},
		\end{align}
	\end{subequations}
	where 
	\begin{equation}
		\alpha_{1,1}=\beta-\frac{1}{\beta},\quad \alpha_{3,1}=2\beta-\frac{1}{\beta},\quad \alpha_{1,3}=\beta-\frac{2}{\beta}.
	\end{equation}
	Finally, we apply odd bosonization from Lemma \ref{iso}
	and get Uglov symmetric functions up to constant factor
	\begin{subequations}
		\begin{align}
	    	\chi_{1,1}&\sim p_1 \sim P_{\left(1\right)}^{1/\beta^2,2},
		    \\
	    	\chi_{3,1} &\sim \frac{2}{3}p_3 + \frac{1}{3}p_1^3+\beta^{-2}p_2p_1\sim P_{\left(3\right)}^{1/\beta^2,2}, 
	    	\\ 
	    	\chi_{1,3}&\sim 2p_3-3p_2 p_1+ p_1^3\sim P_{\left(1,1,1\right)}^{1/\beta^2,2}.
		\end{align}
	\end{subequations}
\end{example}
\begin{example}[R sector] Consider  the highest weights given by formula \eqref{lamrs} with \(rs\leq 2\)
	\begin{equation}
	   \lambda_{2,1}=\frac{\varepsilon}{2\sqrt{2}}\left(2\beta-\frac{1}{\beta}\right),\quad \lambda_{1,2}=\frac{\varepsilon}{2\sqrt{2}}\left(\beta-\frac{2}{\beta}\right).
	\end{equation}
	The corresponding singular vectors in Verma modules have the form
	\begin{equation} \label{eq:chi21}
	   \chi_{2,1}=\left(L_{-1}-\varepsilon\frac{\beta}{\sqrt{2}}G_{-1}\right)\ket{\Delta_{2,1}}, \quad \chi_{1,2}=\left(L_{-1}+\varepsilon\frac{1}{\sqrt{2} \beta}G_{-1}\right)\ket{\Delta_{1,2}}.
	\end{equation}
	After bosonization by Lemma \ref{vir} we get 
	\begin{subequations}
		\begin{align} 
		    \chi_{2,1}&\mapsto \left(\beta-\beta^{-1}\right)\left(a_{-1}-\varepsilon\sqrt{2}\beta f_{-1}\right)\ket{\alpha_{2,1}}, 
		    \\ \chi_{1,2}&\mapsto \left(\beta-\beta^{-1}\right)\left(a_{-1}+\varepsilon\sqrt{2}\beta^{-1} f_{-1}\right)\ket{\alpha_{2,1}}.
		\end{align}
	\end{subequations}
   	Finally, we apply odd bosonization from Lemma \ref{iso} 
    and get Uglov symmetric functions up to constant factor
    \begin{equation}
        \chi_{2,1}\sim p_2+\beta^2 p_1^2\sim P_{\left(2\right)}^{1/\beta^2,2}, \quad \chi_{1,2}\sim -p_2+p_1^2\sim P_{\left(1,1\right)}^{1/\beta^2,2}.
    \end{equation}
\end{example}

\begin{remark}	\label{Rem:sing vectors}
	As we discussed in Remark \ref{Rem:M:tilde} there is another version of Verma modules in Ramond sector. As was explained, if \(\Delta \neq c/24 \) there is an isomorphism  \(\widetilde{M}(c,\Delta)=M(c,\lambda)\oplus M(c,-\lambda)\). The Theorem \ref{Th:sing:NSR} gives formulas for singular vectors in \(M(c,\lambda), M(c,-\lambda)\) and this can be used for the singular vectors on \(\widetilde{M}(c,\Delta)\).

	To be more precise, let $\chi_{r,s}^+$ be a singular vector in $M\left(c,\lambda_{r,s}\right)$  and $\chi_{r,s}^-$ in $M\left(c,-\lambda_{r,s}\right)$. There is an automorphism $\sigma$ on NSR algebra that acts on generators $\sigma\left(G_k\right) = - G_k$,  $\sigma(L_n)=L_n$. Clearly, it changes sign of \(\lambda\), hence permute formulas for \(\chi_{r,s}^+\) and \(\chi_{r,s}^-\). Therefore, we can write 	
	\begin{equation}
		\chi_{r,s}^+=\left(\mathcal{D}_{r,s}^0+\mathcal{D}_{r,s}^1\right)\ket{\Delta_{r,s},\lambda_{r,s}}, \quad \chi_{r,s}^-=\left(\mathcal{D}_{r,s}^0-\mathcal{D}_{r,s}^1\right)\ket{\Delta_{r,s},-\lambda_{r,s}},
	\end{equation}
	where operator \(\mathcal{D}_{r,s}^0\) is even (i.e. invariant under \(\sigma\)) and \(\mathcal{D}_{r,s}^1\) is odd.
%
%
%
	Then, the even and odd singular vectors in $\widetilde{M}\left(c,\Delta_{r,s}\right)$ have the form
%
	\begin{subequations}
		\begin{align} \label{eq:chi:even}
			\widetilde{\chi}_{r,s}^{even}&=\frac{\chi_{r,s}^+ - \chi_{r,s}^-}{2 \lambda_{r,s}}=\left(\mathcal{D}^0_{r,s}+\mathcal{D}^1_{r,s}\frac{G_0}{\lambda_{r,s}}\right)\ket{\Delta^+_{r,s}},
			\\		\label{eq:chi:odd}
			\widetilde{\chi}_{r,s}^{odd}&=\frac{\chi_{r,s}^+ + \chi_{r,s}^-}{2}=\left(\mathcal{D}^0_{r,s}+\mathcal{D}^1_{r,s}\frac{G_0}{\lambda_{r,s}}\right)\ket{\Delta^-_{r,s}},
		\end{align}
	\end{subequations}
	where the relation between $\ket{\Delta^{\pm}_{r,s}}$ and \(\ket{\Delta_{r,s},\pm\lambda_{r,s}}\) was given in the formulas \eqref{altR}.
\end{remark}	
\begin{example} Let \((r,s)=(2,1)\). Then combining formulas \eqref{eq:chi21} and \eqref{eq:chi:even}-\eqref{eq:chi:odd} we get 
	\begin{equation}
		\widetilde{\chi}_{2,1}^{even}=\left(L_{-1}-\frac{2\beta^2}{2\beta^2-1}G_{-1}G_0\right)\ket{\Delta_{2,1}^+}, \quad  \widetilde{\chi}_{2,1}^{odd}=\left(L_{-1}-\frac{2\beta^2}{2\beta^2-1}G_{-1}G_0\right)\ket{\Delta_{2,1}^-}.
	\end{equation}
\end{example}

\subsection{Proof of the Theorem \ref{Th:sing:NSR}}
We should prove the following properties of Uglov polynomial 
\begin{equation}\label{eq:ToProve}
	G_kP^{1/\beta^2,2}_{\left(r^s\right)}=0,\ k\in \mathbb{Z}_{>0} + \delta,\quad L_n P^{1/\beta^2,2}_{\left(r^s\right)}=0,\ n\in \mathbb{Z}_{>0}, \quad P^{1/\beta^2,2}_{\left(r^s\right)}\in F_{\alpha_{r,s}}.
\end{equation}
It follows from Lemma \ref{Maclem} that  \(P_{\lambda}\l(-e^{\hbar},-e^{\gamma \hbar}\r) = P_{\lambda}\l(-e^{-\hbar},-e^{-\gamma \hbar}\r)\). Therefore, we have 
\begin{equation}
    P_{(r^s)}\l(q,t\r)= P_{(r^s)}^{\gamma,2}+O\l(\hbar^2\r).
\end{equation}
Recall the formula for the \(\hbar\) expansion of \(T(z)\)
\begin{equation}
   q^{-1/2}t^{1/2}u^{-1} T\l(z\r)=T^0\l(z\r)+T^1\l(z\r)\hbar+O\l(\hbar^2\r).
\end{equation}
Using these series expansions in the formula \eqref{eq:Tn:Macdonald} we get (for any \(n>0\)).
\begin{equation}
\label{TnP}
    \l(T^0_n+T^1_n \hbar +O\l(\hbar^2\r) \r)\l( P_{(r^s)}^{\gamma,2}+O\l(\hbar^2\r)\r) =0 \Rightarrow \left\{ \begin{aligned} T^0_nP_{(r^s)}^{\gamma,2}=0, \\
        T^1_n P_{(r^s)}^{\gamma,2}=0. 
    \end{aligned}\right. 
\end{equation}
Hence, using Theorems \ref{T0T1NS} and \ref{T0T1R} with \(\gamma=1/\beta^2\) we obtain
\begin{gather}
    f_n^{add}P^{1/\beta^2,2}_{\left(r^s\right)}=0, \;\; \forall n>0, \label{eq:f:Uglov} \\
    G_k P^{1/\beta^2,2}_{\left(r^s\right)}=0, \;\; \forall k>0. \label{eq:G:Uglov}
\end{gather}
\begin{lemma}\label{Lem:Fadd}
    For any element $v\in \Lambda \cong F_{\alpha}\otimes F^{add}$ the condition \(f^{add}_k v=0\) \(\forall k>0\) is equivalent to 
        \(v\in F_{\alpha}\otimes \ket{1}\)
\end{lemma}
\begin{proof} 
	Clearly \(v\) contains some terms with \(f_{-k}^{add}\) if and only if \(f_k^{add} v \neq 0 \). 
\end{proof}
Using the Lemma \ref{Lem:Fadd} we conclude that 
\(   P^{1/\beta^2,2}_{\left(r^s\right)}\in F_{\alpha_{r,s}}\otimes \ket{1}.
\)

Now we consider two sectors separately. In the Neveu-Schwarz sector we use \(\{G_{1/2},G_{n-1/2}\}=2L_n\) for \(n>0\) and conclude from \eqref{eq:G:Uglov} that \( L_n P^{1/\beta^2,2}_{\left(r^s\right)}=0, n>0\). Hence, properties \eqref{eq:ToProve} are proven.
%
%

In the Ramond sector we will additionally use property 
\begin{equation}\label{eq:G0:Uglov:tilde}
	G_0 P^{1/\beta^2,2}_{\left(r^s\right)} \sim P^{1/\beta^2,2}_{\left(r^s\right)}.
\end{equation}
which we will prove in Section \ref{sec6}. Using the formula \(\{G_{0},G_{n}\}=2L_n\) for \(n>0\) and relation~\eqref{eq:G:Uglov} we conclude that \( L_n P^{1/\beta^2,2}_{\left(r^s\right)}=0, n>0\). Hence, properties \eqref{eq:ToProve} are proven.~\(\square\)
%
%

\section{Limit of Macdonald operator}
\label{sec6}
In this section we study the limit of Macdonald operator and it's eigenvalues. Recall the formula~\eqref{eq:Macd:oper:N} for the operator \(D_{q,t}\) acting on \(\Lambda_{N}\). We will use the following formula for bosonization of Macdonald operator. 
\begin{fact}
    Let \(\pi_N\colon \Lambda \rightarrow \Lambda_N\) denotes natural projection. Then the operator \(\mathrm{D}_{q, t}\) on  \(\Lambda\) given by the formula
    \begin{equation}
        \mathrm{D}_{q, t}=\frac{t^{N}}{t-1} \oint \frac{dz}{2\pi \ri z} \eta(z)-\frac{1}{t-1},
    \end{equation}
    where 
    \begin{equation}
        \eta(z)=\exp \left(\sum_{n>0}\left(1-t^{-n}\right) \frac{p_{n}}{n} z^{n}\right) \exp \left(-\sum_{n>0}\left(1-q^{n}\right) \frac{\partial}{\partial p_{n}} z^{-n}\right),
    \end{equation}
	satisfies \(\pi_N \mathrm{D}_{q,t} = D_{q,t} \pi_N\).
\end{fact}
See e.g. \cite[App. B]{AKOS:1996} for the proof of this formula. It follows from the formula for eigenvalue in Theorem \ref{Th:Macd:polyn} that
\begin{equation}\label{eq:Dqt:Macdonald}
	\oint \frac{dz}{2\pi \ri z} \eta(z) P_{\mu}\l(q,t\r)=\mathcal{E}_{\mu}\l(q,t\r)P_{\mu}\l(q,t\r), \quad \mathcal{E}_{\mu}(q, t)=1+(t-1) \sum_{i=1}^{\ell(\mu)}\left(q^{\mu_{i}}-1\right) t^{-i}.
\end{equation}

Consider the $\hbar\rightarrow 0$ limit of this equation, where as before $q=-e^{\hbar}$, $t=-e^{\gamma \hbar}$. Macdonald symmetric function becomes Uglov symmetric function $P_{\mu}\l(q,t\r)\rightarrow P^{\gamma,2}_{\mu}+O\l(\hbar^2\r)$ and 
\begin{equation}
	\eta(z)=C^0(z)+ \hbar C^1(z)+O\l(\hbar^2\r), \quad \mathcal{E}_{\mu}(q, t) = \mathcal{E}^0_{\mu}+\hbar\mathcal{E}^1_{\mu}+O\l(\hbar^2\r).
\end{equation}
Following \cite{SKAO}, we can express Macdonald operator through modes of generating current of $Vir_{q,t}$
\begin{equation}\label{op}
	\oint \frac{dz}{2 \pi \ri z} \eta(z)=\sum\limits_{n=0}^{\infty} \psi_{-n}T_{n}-q^{-1}tu^{-2},
\end{equation}
where $\psi_n$ are modes of the following operator
\begin{equation}\label{psi}
	\psi(z)=\sum_{n=0}^{\infty} \psi_{-n} z^{n}=q^{-1 / 2}t^{1/2} \exp \left\{-\sum_{n=1}^{\infty} \frac{1-t^{n}}{t^n+q^{n}}\frac{q^{n/2}}{t^{n/2}} \frac{p_n}{n} z^{n}  \right\} u^{-1}.
\end{equation}
Let \(\psi(z)\) has the following $\hbar$ expansion
\begin{equation}
        q^{1/2}t^{1/2}u\psi(z)=\psi^0(z)+\psi^1(z)\hbar+O\left(\hbar^2\right).
\end{equation}
In the zero \(\hbar \) order of Macdonald operator and eigenvalue  \eqref{eq:Dqt:Macdonald} we get
\begin{equation}
	\oint \frac{dz}{2 \pi \ri z} C^0(z)=\sum\limits_{n=0}^{\infty} \psi^0_{-n}T^0_{n}\mp 1,\quad \mathcal{E}^0_{\mu}=1-2 \sum_{i=1}^{\ell(\mu)}\big((-1)^{\mu_{i}}-1\big) (-1)^{i}.
\end{equation}
In the first \(\hbar\) order we get
\begin{equation}
	\oint \frac{dz}{2\pi \ri z}C^1(z)=\sum_{n=0}^{\infty}\l( \psi^0_{-n} T^1_{n}+\psi^1_{-n} T^0_{n}\r)\pm\l(1-2 \beta^{-1}\alpha -\gamma\r),
\end{equation}
\begin{equation}\label{eq:E1}
	\mathcal{E}_{\mu}^{1}(\gamma)=-\sum\limits_{i=1}^{\ell(\mu)}(-1)^{i}\Big( 2(-1)^{\mu_{i}} \mu_{i}+\gamma(1-2 i)\big((-1)^{\mu_{i}}-1\big)\Big).
\end{equation}
Here the upper sign corresponds to NS sector and the lower sign corresponds to  R sector. Let us study eigenvalues of these operators on singular vector which is represented by Uglov symmetric function $P_{\left(r^s\right)}^{1/\beta^2, 2}$. We consider each sector separately.

\noindent \textbf{Neveu-Schwarz sector.} Using Theorem \ref{T0T1NS} and formulas \eqref{eq:f:Uglov}-\eqref{eq:G:Uglov} we have
\begin{subequations}
	\begin{align}
        \oint \frac{dz}{2\pi \ri z}C^0(z)P^{1/\beta^2,2}_{(r^s)} &= \l(\varepsilon 2 \sqrt{2}f^{R}_0-1\r)P^{1/\beta^2,2}_{(r^s)},
        \\
        \oint \frac{dz}{2\pi \ri z}C^1(z)P^{1/\beta^2,2}_{(r^s)} &= \l(\varepsilon\sqrt{2}f_0^{R}-\varepsilon\sqrt{2}f^{R}_0\l(2-2\alpha_{r,s}\beta^{-1}-\beta^{-2}\r)+1-2\alpha_{r,s}\beta^{-1}-\beta^{-2}\r)P^{1/\beta^2,2}_{(r^s)}.
	\end{align}
\end{subequations}
Let $\widehat{\Pi}$ denotes parity operator on Fock module \(\Lambda \cong \mathrm{F}_{\alpha,\varepsilon}\), it acts as \(1\) on vectors with even number of fermions and as \(-1\) on vectors with odd number of fermions. It is easy to see that $f^R_0$ acts on Fock module as $\frac{\varepsilon }{\sqrt{2}}\widehat{\Pi}$. Clearly \(P^{1/\beta^2,2}_{(r^s)}\) is even if both \(r,s\) are even and \(P^{1/\beta^2,2}_{(r^s)}\) is odd if both \(r,s\) are odd. 

Now we can see agreement between expansion of the eigenvalue in \(\mathcal{E}_{(r^s)}\) and our operator formulas. In zero order 
\begin{itemize}
	\item If $r,s$ are even we have $\oint \frac{dz}{2\pi \ri z}C^0(z)P^{1/\beta^2,2}_{(r^s)}=1 P^{1/\beta^2,2}_{(r^s)}$ which agrees with $\mathcal{E}^0_{\left(r^s\right)}=1$,
	\item If $r, s$ are odd we have $\oint \frac{dz}{2\pi \ri z}C^0(z)P^{1/\beta^2,2}_{(r^s)}=-3 P^{1/\beta^2,2}_{(r^s)}$ which agrees with $\mathcal{E}^0_{\left(r^s\right)}=-3$.
 \end{itemize}
In first order 
\begin{itemize}
	\item If $r,s$ are even we have $ \oint \frac{dz}{2\pi \ri z}C^1(z)P^{1/\beta^2,2}_{(r^s)}=0$ which agrees with $\mathcal{E}^1_{\left(r^s\right)} = 0 $,
	\item If $r, s$ are odd we have $ \oint \frac{dz}{2\pi \ri z}C^1(z)P^{1/\beta^2,2}_{(r^s)}=\left(2-4\beta^{-1}\alpha_{r,s}-2\beta^{-2} \right)P^{1/\beta^2,2}_{(r^s)}$ which agrees with $\mathcal{E}^1_{\left(r^s\right)} = \left(-2r+2\beta^{-2} s\right)$.
\end{itemize}

\noindent  \textbf{Ramond sector.} Using Theorem \ref{T0T1R} and formulas \eqref{eq:f:Uglov}-\eqref{eq:G:Uglov} we have
\begin{subequations}
    \begin{align}
	    \label{Reig}
        \oint \frac{dz}{2 \pi \ri z} C^0(z)P^{1/\beta^2,2}_{(r^s)} &= \left(T^0_{0}+1\right)P^{1/\beta^2,2}_{(r^s)}=P^{1/\beta^2,2}_{(r^s)},
		\\
    	\label{RC1}
        \l(\oint \frac{dz}{2\pi \ri z}C^1(z)\r) P^{1/\beta^2,2}_{(r^s)} &=  \l(-\varepsilon 2\sqrt{2}\beta^{-1} G_0-\l(1-2 \beta^{-1}\alpha_{r,s} -\gamma\r)\r)P^{1/\beta^2,2}_{(r^s)}.
    \end{align}
\end{subequations}
Note that in the formula \eqref{RC1} the mode \(G_0\) is bosonized through fermion \(f^R\). 
Formula \eqref{Reig} agrees with $\mathcal{E}^0_{\left(r^s\right)}=1$ for $r$ and $s$ of different parity. Now agreement between the operator formula \eqref{RC1} and the eigenvalue formula  \eqref{eq:E1} leads to 
\begin{equation}\label{eq:G0:Uglov}
	G_0 P^{1/\beta^2,2}_{(r^s)}=(-1)^s\varepsilon \frac{r\beta+s\beta^{-1}}{2\sqrt{2}} P^{1/\beta^2,2}_{(r^s)}.
\end{equation}
This was used in the proof of Theorem \ref{Th:sing:NSR}, see relation \eqref{eq:G0:Uglov:tilde}. Together with the statement of Theorem \ref{Th:sing:NSR} this gives a new proof of Proposition \ref{sign}. 

Indeed, substituting expression \eqref{eq:G0:Uglov} to the formula \eqref{RC1} we get 
\begin{equation}
    \label{RC2}
        \l(\oint \frac{dz}{2\pi \ri z}C^1(z)\r) P^{1/\beta^2,2}_{(r^s)}=\left(\left(-1\right)^{s+1}\left(r+s\beta^{-2}\right)+\left(r-s\beta^{-2}\right)\right)P^{1/\beta^2,2}_{(r^s)}.
\end{equation}
And this agrees with formula \eqref{eq:E1} for \(\mathcal{E}_{\left(r^s\right)}^{1}\) 
\begin{itemize}
	\item If $r$ is even and $s$ is odd we have $\l(\oint \frac{dz}{2\pi \ri z}C^1(z)\r) P^{1/\beta^2,2}_{(r^s)}=2r P^{1/\beta^2,2}_{(r^s)}$ which agrees with $\mathcal{E}_{\l(r^s\r)}^{1}=-\sum\limits_{i=1}^{s}(-1)^{i}2 r =2r$,
    \item If $r$ is odd and $s$ is even we have $\l(\oint \frac{dz}{2\pi \ri z}C^1(z)\r) P^{1/\beta^2,2}_{(r^s)}=-2s\beta^{-2} P^{1/\beta^2,2}_{(r^s)}$ which agrees with $\mathcal{E}_{\l(r^s\r)}^{1}=\sum\limits_{i=1}^{s}(-1)^{i}\l(\beta^{-2}\l(1-2i\r)2\r)=-2s\beta^{-2}$.
\end{itemize}


    \bibliographystyle{alpha}
        \bibliography{references}
        
		\noindent \textsc{Landau Institute for Theoretical Physics, Chernogolovka, Russia,\\
	Center for Advanced Studies, Skoltech, Moscow, Russia,\\
	National Research University Higher School of Economics, Moscow, Russia}

\emph{E-mail}:\,\,\textbf{mbersht@gmail.com}\\

\noindent\textsc{National Research University Higher School of Economics, Moscow, Russia}

\emph{E-mail}:\,\,\textbf{angel\_1997@inbox.ru}	        
\end{document}